\documentclass[journal]{IEEEtran}
\usepackage[utf8]{inputenc}
\usepackage{times}
\usepackage{graphicx}
\usepackage{amsmath}
\usepackage{amsfonts}
\usepackage{amsthm}
\usepackage{amssymb}
\usepackage{nicefrac}
\usepackage{psfrag}
\usepackage{multirow}
\usepackage{wasysym}
\usepackage{enumerate}
\usepackage[dvipsnames,usenames]{xcolor}
\usepackage{booktabs}[]
\usepackage{multirow}
\usepackage{soul}

\newtheorem{lemma}[]{Lemma}
\newtheorem{theorem}[]{Theorem}

\newcommand{\Exp}{{\mathbf{E}}}
\newcommand{\dd}{{\rm d}}

\newcommand{\tr}{{\rm tr}}
\newcommand{\diag}{{\rm diag}}
\newcommand{\RR}{\mathbb{R}}

\newcommand{\xok}{x^{\rm ok}}
\newcommand{\xko}{x^{\rm ko}}
\newcommand{\xh}{\hat{x}}

\newcommand{\xhok}{\xh^{\rm ok}}
\newcommand{\xhko}{\xh^{\rm ko}}

\newcommand{\cov}{{\bf \Sigma}}
\newcommand{\covok}{\cov^{\rm ok}}
\newcommand{\covko}{\cov^{\rm ko}}
\newcommand{\lambdaok}{\lambda^{\rm ok}}
\newcommand{\lambdako}{\lambda^{\rm ko}}
\newcommand{\fok}{f^{\rm ok}}
\newcommand{\fko}{f^{\rm ko}}

\newcommand{\PP}{{\mathbb{S}}}
\newcommand{\OO}{{\mathbb{O}}}
\newcommand{\fenc}{\fxhgx^\ok}

\newcommand{\fokx}{f^{\rm ok}_x}
\newcommand{\fkox}{f^{\rm ko}_x}
\newcommand{\fokxh}{f^{\rm ok}_{\xh}}
\newcommand{\fkoxh}{f^{\rm ko}_{\xh}}

\newcommand{\fxhgx}{f_{\xh|x}}

\newcommand{\fkoxhjx}{\fko_{\xh,x}}

\newcommand{\covokxh}{\hat{\cov}^{\rm ok}_\theta}
\newcommand{\covkoxh}{\hat{\cov}^{\rm ko}_\theta}
\newcommand{\covXth}{\tilde{\cov}_\theta}

\newcommand{\II}{{\mathcal{I}}}

\newcommand{\Unif}[1]{{\mathcal U}\left(#1\right)}
\newcommand{\Gauss}[2]{{\mathcal{G}}\left(#1,#2\right)}
\newcommand{\GPDF}[3]{G_{#2,#3}\left(#1\right)}

\newcommand{\DKL}[2]{\mathcal{D}_{\rm KL}\left(#1\Vert#2\right)}

\newcommand{\Lambdako}{\Lambda^{\rm ko}}
\newcommand{\Uko}{U^{\rm ko}}

\newcommand{\Tth}{T_\theta}
\newcommand{\Sth}{S_\theta}
\newcommand{\nth}{n_\theta}
\newcommand{\uokthj}{u_{\theta,j}}

\newcommand{\jb}{\bar{\jmath}}
\newcommand{\kb}{\bar{k}}

\newcommand{\Loc}{{\mathcal L}}

\newcommand{\AUC}{{\rm AUC}}
\newcommand{\ok}{{\rm ok}}
\newcommand{\ko}{{\rm ko}}

\newcommand{\ellxhok}{\ell_{\xh}^{\ok}}

\title{Anomaly Detection based on Compressed Data: an Information Theoretic Characterization}

\author{Alex~Marchioni,~\IEEEmembership{Student Member,~IEEE,}
        Andriy Enttsel,~
        Mauro~Mangia,~\IEEEmembership{Member,~IEEE,}
        Riccardo~Rovatti,~\IEEEmembership{Fellow,~IEEE,}
        and~Gianluca~Setti,~\IEEEmembership{Fellow,~IEEE}
\thanks{This work has been submitted to the IEEE for possible publication. Copyright may be transferred without notice, after which this version may no longer be accessible.}
\thanks{A. Marchioni, A. Enttsel, M. Mangia and R. Rovatti are with the Department of Electrical, Electronic, and Information Engineering, University of Bologna, 40136 Bologna, Italy, and also with the Advanced Research Center on Electronic Systems, University of Bologna, 40125 Bologna, Italy (e-mail: alex.marchioni@unibo.it, andriy.enttsel@unibo.it, mauro.mangia@unibo.it, riccardo.rovatti@unibo.it).}%
\thanks{R. Rovatti is also with the Alma Mater Research Institute for Human Centered AI University of Bologna, 40015 Bologna, Italy}
\thanks{G. Setti is with the Department of Electronics and Telecommunications, Politecnico di Torino, 10129 Torino, Italy, and also with the Advanced Research Center on Electronic Systems (ARCES), University of Bologna, 40125 Bologna, Italy (e-mail: gianluca.setti@polito.it).}%
}

\begin{document}

\maketitle
\begin{abstract}
We analyze the effect of lossy compression in the processing of sensor signals that must be used to detect anomalous events in the system under observation.
The intuitive relationship between the quality loss at higher compression and the possibility of distinguishing anomalous behaviours from normal ones is formalized in terms of information-theoretic quantities.
Some analytic derivations are made both in a Gaussian framework and in the asymptotic case for what concerns the extent of signals considered.

Analytical conclusions are matched with the performance of practical detectors in a simple case allowing the assessment of different compression/detector configurations.
\end{abstract}


\section{Introduction}

A typical scenario for nowadays massive acquisition systems can be modelled as a large number of sensing units, each transforming some physical unknown quantity into samples of random processes that are then transmitted over a network.
To reduce transmission bitrate, signals are often compressed by a lossy mechanism that is theoretically capable of preserving useful information. Before reaching some cloud facility in which they will be ultimately stored or processed, the corresponding bitstreams may traverse several levels of hierarchical aggregation and intermediate devices that are often indicated as the {\em edge} of the cloud \cite{Shi_JIOT2016}.
For latency or privacy reasons, some computational tasks may benefit from their deployment at the edge. One of those tasks is the detection of anomalies/novelties.

This is especially true when dealing, for example, with networks that sensorize plants or structures subject to monitoring as depicted in Fig.~\ref{fig:overall}. The aggregated sensor readings may be processed in the cloud for off-line monitoring relying on long-term historical trends, while the outputs of subsets of sensors may be processed at the edge to give low-latency feedback on possible critical events that require immediate intervention.

Usually, compression schemes applied to sensor data are asymmetric and entail a lightweight encoding performed on very-low complexity devices paired with a possibly expensive decoding stage running on the cloud. In these conditions, it is sensible that anomaly detectors work on compressed data and not on the recovered signal.

Yet, lossy compression bases its effectiveness on neglecting some of the signal details. This translates into a distortion between the original and the recovered signal but also in a loss of details that, in principle, could have been used to tell normal behaviours from anomalous ones.

In general, acquisition systems must obey a distortion constraint so that they are designed to best address the trade-off between compression and distortion. However, such a trade-off goes in parallel to the one between distortion and the ability to determine if the signal is normal or anomalous. Here, we analyze the latter with the same information theoretic machinery used in the well-known rate-distortion analysis and implicitly show that the two trade-offs are different.

\begin{figure}
    \centering
    \includegraphics[width=\columnwidth]{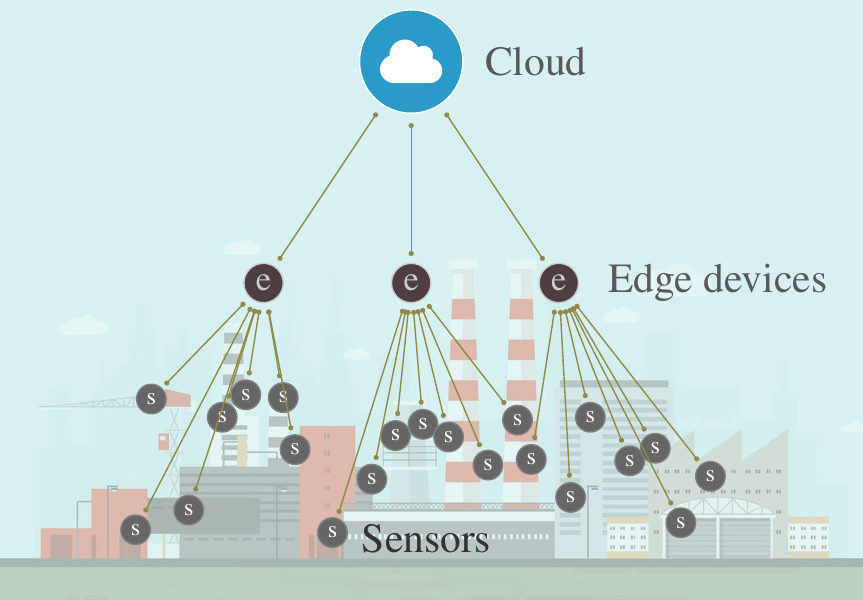}
    \caption{\label{fig:overall}A sensorized plant whose acquisitions are aggregated at the edge before being sent to the cloud.}
\end{figure}

How compression affects distinguishability has been investigated in the literature.  In \cite{Ahlswede_TIT1986} the problem of hypothesis testing is discussed for a single source under a rate constraint. Such a basis has been extended to information-theoretic problems of statistical inference in the case of multiterminal data compression in \cite{Han_TIT1998}. These works address the inference problem with no constraint on distortion since compressed data is not required to be recoverable.
On the contrary, we assume that compression is designed to guarantee the quality of service needed by the processing tasks that receive the reconstructed data.

In a sense, the framework we identify is somehow similar to the information-bottleneck scheme \cite{Tishby_Allerton1999, Slonim_NIPS1999}. In that scheme, distortion is replaced by a very general criterion which identifies the features that should be preserved when compressing with the information that the original signal contains about a second (suitably introduced) signal. Yet, our discussion takes a different direction as, when dealing with anomaly/novelty/outlier, we may completely ignore the statistics of the anomaly and, even if we have priors on that, we need to be able to treat also cases in which the mutual information between normal and anomalous signals is null.

This information-bottleneck principle has also been employed for unsupervised tasks. To tackle one-class classification, in \cite{Crammer_ICML2004, Crammer_ICML2008} an optimization problem is considered in terms of rate-distortion trade-off and solve it by applying the information-bottleneck principle. However, this trade-off is used for anomaly detection with no consideration about distortion.

Another way to relate unsupervised anomaly detection to data compression is described in \cite{Bohm_KDD2009}. Here, the level of abnormality of a data point with respect to the entire data set, called {\em coding cost}, is given by its ability to be efficiently compressed in a Huffman coding fashion.

The same reason that differentiates our work from the information-bottleneck principle makes the analysis we propose different from other modifications of classical rate-distortion theory that substitute energy-based distortion with perceptive criteria \cite{Blau_CVPR2018, Blau_ICML2019}.

Though not overlapping with the problem we address, it is also worthwhile to mention \cite{Rodrigues_ICASSP2017, Shlezinger_TSP2019}, in which it is assumed that the original signal is characterized by some parameters (e.g., their mean) and study how the estimation of such parameters is affected by lossy compression.

Note also that other applications exist in which rate and distortion are paired with additional merit figures taking into account relevant features of the system. As an example \cite{Foo_ICASSP2007} adds computational effort considerations to the analysis of rate-distortion of wavelet-based video coding.

Finally, even without emphasis on compression, the relationship between the analysis of suitably defined subcomponents of a signal to detect possible outlier behaviours is a classic theme that is still under investigation \cite{Menon_TSP2019, Kong_TSP2019, Rahmani_ICASSP2021}.

In this paper, we propose an analysis of the performance of a generic detector working on a signal distorted by the compression mechanism that minimizes the rate given a constraint on distortion. To characterize the detector, we define two information-theoretic measures of distinguishability to model anomaly-agnostic and anomaly-aware scenarios. Specializing our analysis to the case of Gaussian sources, we show that: \textit{i}) the distinguishability metrics in case of white anomalies are representative of the average performance of the detector evaluated on many other different anomalies; \textit{ii}) as the signal dimension increases, any possible anomaly tends to resemble the white anomaly; \textit{ii}) in case of an anomaly-agnostic detector, there exists at least one critical distortion level that makes the white anomaly undetectable.

The paper is organised as follows. Section~\ref{sec:ratedistortion} reviews the classical rate-distortion theory in a general setting first and then in the specific case of Gaussian sources with considerations on the optimal encoding mapping expression.
Section~\ref{sec:distinguishability} provides the definition of the normal and anomalous signals together with the formulation of the distinguishability measures for both anomaly-agnostic and anomaly-aware scenarios. 
Section~\ref{sec:averagelarge} focuses on the distinguishability in the average case, with an emphasis on the asymptotic characterization of high-dimensional signals.
Section \ref{sec:numerical} reports some numerical evidence analysing the behaviour of some suitably simplified anomaly detection strategies with respect to ideal and suboptimal compression strategies. Theoretical curves anticipate many aspects of practical performance trends and show that compression that optimizes the rate-distortion trade-off is not necessarily addressing at best the compromise with distinguishability. The conclusion is finally drawn. Proofs of the theorems and lemmas stated in the discussion are reported in the Appendix.

\section{Rate vs. distortion}
\label{sec:ratedistortion}

\begin{figure}
    \centering
    \includegraphics[width=\columnwidth]{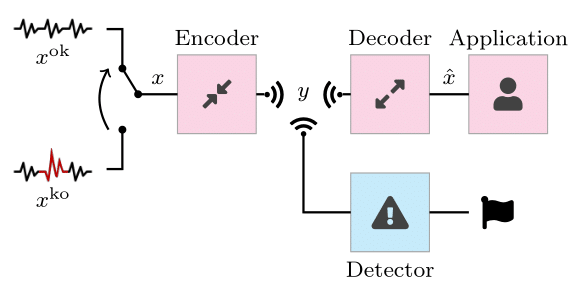}
    \caption{ The signal chain is tuned on the normal signal $\xok$ to best address the rate-distortion trade-off, guaranteeing a certain quality of service to a given application.
    An anomalous signal $\xko$ may occur and a detector working on the compressed signal $y$ should be able to detect it.
    }
    \label{fig:sigchain}
\end{figure}

We consider the context in which a system has the main task of transferring the information content of a signal source $x$ to a receiver through a communication channel that has a constraint on rate. At any time instant $t$, an instance $x[t]$ is passed to an encoding stage producing a compressed version $y[t]$ that may then be decompressed into $\xh[t]\in \widehat{\RR^n}$, where $\widehat{\RR^n} \subset \RR^n$.

The constraint on rate is such that it implies a lossy compression mechanism. The encoding stage is therefore not injective and introduces some distortion. The encoder is tuned on the source $x$, which is modelled as a independent discrete-time, $n$-dimensional stochastic processes.

The trade-off between rate and distortion is addressed in the rate-distortion theory \cite[Chapter 13]{Cover_1991}. Distortion may be defined as
\begin{equation}
\label{eq:distortiondef}
D = \Exp\left[\left\|x[t]-\xh[t]\right\|^2\right]
\end{equation}
where $\Exp[\cdot]$ stands for expectation, and the minimal achievable rate $\rho$ can be expressed as a function of the maximal accepted distortion $\delta$  as follows \cite[Theorem 13.2.1]{Cover_1991}
\begin{equation}
\label{eq:rate-distortion}
\rho(\delta)=\inf_{\fxhgx} I\left(\xh;x\right) \quad\text{s.t. $D\le \delta$}
\end{equation} 
where $I\left(\xh;x\right)$ is the mutual information between $\xh$ and $x$ \cite[Chapter 8]{Cover_1991}, and $\fxhgx$ is a conditional probability density function (PDF) modeling the possibly stochastic mapping characterizing the encoder-decoder pair. Although \cite[Theorem 13.2.1]{Cover_1991} defines the rate-distortion function in the discrete case, it can also be proved for well-behaved continuous sources \cite[Chapter 13]{Cover_1991} as considered in this work.

If the source is memoryless (thus allowing us to drop the time index $t$) and generates vectors of independent and zero-mean Gaussian variables, i.e., when $x\sim\Gauss{0}{\cov}$ where $\cov$ is a diagonal covariance matrix such that $\cov = \diag(\lambda_0, \dots, \lambda_{n-1})$ with $\lambda_0 \ge \lambda_1 \ge \dots \ge \lambda_{n-1} \ge 0$, then the solution of \eqref{eq:rate-distortion} is

\begin{align}
\label{eq:GRD-rate}
\rho 
&= \frac{1}{2}\sum_{j=0}^{n-1}\log_2\frac{\lambda_j}{\min\left\{\theta,\lambda_j\right\}}
=  - \frac{1}{2}\sum_{j=0}^{n-1}\log_2 \tau_j \\
\label{eq:GRD-delta}
\delta 
&= \sum_{j=0}^{n-1} \min\left\{\theta,\lambda_j\right\}
= \sum_{j=0}^{n-1} \lambda_j \tau_j
\end{align}
where $\theta\in[0,\lambda_0]$ is the so called reverse water-filling parameter \cite[Theorem 13.3.3]{Cover_1991}, and $\tau_j = \min\left\{1,\theta/\lambda_j\right\}$ accounts for the fraction of energy cancelled by distortion along the $j$-th component. 

The coding theorems behind such a classical development imply that the optimal trade-off \eqref{eq:rate-distortion} between rate and distortion is asymptotically obtained by simultaneously encoding an increasing number of subsequent source symbols into a single block that can be then reverted to a sequence of distorted symbols. Hence, in principle, the intermediate symbols $y$ feeding the anomaly detector in Fig.~\ref{fig:sigchain} cause it to work simultaneously on multiple instances of the signals.

Though this is not incoherent with what happens in real detectors that observe more than one suspect instance before declaring an anomaly, we here instead consider a {\em per-use} analysis which is typical and scales the key merit figures (rate, distortion and, in our case, distinguishability -- see Section \ref{sec:distinguishability}) by the number of source symbols aggregated to obtain them.

This allows us to pursue the classical approach defining a test channel whose single use has the same expected behaviour as the average of infinite uses and, in the case of Gaussian sources, has a particularly simple expression that we derive and exploit to imagine that a source instance $x$ is encoded into a compressed symbol $y$ from which $\xh$ can be recovered \cite[Chapter 13]{Cover_1991}, \cite{Kolmogorov_TIT1956}.

In the same Gaussian framework, it is also possible to derive the PDF of the distorted signal $\xh$ and the conditional PDF $\fxhgx$ that stochastically maps an input $x\sim\Gauss{0}{\cov}$ to $\xh$.
If we accept to identify a zero-variance Gaussian with a Dirac's delta and define $\Sth = I_n - \Tth$ with $\Tth=\diag(\tau_0, \dots, \tau_{n-1})$ to account for the fraction of energy that survives distortion along each component, then we can derive the following Lemma whose proof is in the Appendix.

\begin{lemma}
\label{lem:Gfenc}
If $x\sim\Gauss{0}{\cov}$ is a memory-less source and we constraint the distortion $D \le \delta$, the optimally distorted signal has distribution
\begin{equation}
\label{eq:Gfokxh}
\xh \sim \Gauss{0}{\cov \Sth}
\end{equation}
and the optimal encoding mapping is
\begin{equation}
\label{eq:Gfenc}
\fxhgx(\alpha,\beta) = \GPDF{\alpha}{\beta\Sth}{\cov\Sth\Tth}
\end{equation}
where $\GPDF{\cdot}{m}{K}$ represents the PDF of a Gaussian variable with mean $m$ and covariance matrix $K$.
\end{lemma}

Although in general it is not explicitly reported, the expression of $\fxhgx$ is important when the compression mechanism is employed to encode a signal different from the one for which it was designed. This is the case of an unexpected anomalous source that replaces the normal signal.

\section{Anomalies and their detectability}
\label{sec:distinguishability}

In the path from encoder to decoder, the compressed signal $y$ may be intercepted for some local processing. The local processing we focus on is the task of distinguishing whether the transmitted signal differs from what is usually observed, i.e., anomaly detection.

To include this aspect in our model, each observable instance $x[t]$ has to be considered as a realization of two different sources: one modeling the normal behaviour $\xok$ and one representing an anomaly $\xko$.
These two sources are modelled as two discrete-time, stationary, $n$-dimensional stochastic processes each generating independent and identically distributed (i.i.d.) vectors $\xok\in\RR^n$ and $\xko\in\RR^n$ with different PDFs $\fok:\RR^n\rightarrow\RR^+$ and $\fko:\RR^n\rightarrow\RR^+$. As a result, at any time $t$ the observable process is either $x[t]=\xok[t]$ or $x[t]=\xko[t]$ (visually represented in Fig.~\ref{fig:sigchain}). Since, we assume the generated vectors as i.i.d. from now on we may drop time indication.

More specifically, according to the framework characterizing Lemma~\ref{lem:Gfenc}, we here consider the case in which both sources are Gaussian. In particular we focus on signals with zero-mean and covariance matrices $\covok,\covko\in\RR^{n\times n}$. In general, $\covok\neq\covko$, but we will assume $\tr(\covok)=\tr(\covko)=n$, where $\tr(\cdot)$ stands for matrix trace, meaning that, on the average, each sample in the vector contributes with a unit energy.
With the assumption of signals to be zero-mean and of equal energy, we can focus our analysis on one of the possible effects of anomalies, i.e., the distribution of energy over the signal subspace. 

Moreover, with no loss of generality, we assume $\covok={\rm diag}\left(\lambdaok_0,\dots,\lambdaok_{n-1}\right)$ with $\lambdaok_0\ge\lambdaok_1\ge\dots\ge\lambdaok_{n-1}\ge 0$.

The signal $x$ is encoded with a compression mechanism tailored for the typical condition in which $x=\xok$. The objective consists in guaranteeing a proper quality of service $D = \Exp[\|\xok[t]-\xhok[t]\|^2] \le \delta$ to the final user. Hence, in this specific case, the rate-distortion function in \eqref{eq:GRD-rate} considers $\cov=\covok$ and $\fxhgx = \fenc$.

Simultaneously, a detector observes $y$ for anomaly detection. Since we assume the decoding stage to be injective, $y$ brings the same information of $\xh$ so that, in abstract terms, processing $y$ is equivalent to working on $\xh$. As a result, the detector works on the difference between the two marginal distributions $\fokxh$ and $\fkoxh$ that can be computed as follows
\begin{align}
\fokxh(\alpha) &= \int_{\RR^n} \fenc\left(\alpha,\beta\right) \fokx(\beta) \dd\beta \label{eq:fokxh}\\
\fkoxh(\alpha) &= \int_{\RR^n} \fenc\left(\alpha,\beta\right) \fkox(\beta) \dd\beta \label{eq:fkoxh}
\end{align}

From \eqref{eq:fkoxh}, it is evident that the compression mechanism $\fenc$ that optimally addresses the rate-distortion trade-off for the normal source is used also on the anomalous instances. Under the i.i.d. Gaussian assumption, \eqref{eq:fokxh} reduces to \eqref{eq:Gfokxh} with $\cov = \covok$, while the PDF of $\xhko$ is given by the following Lemma, whose proof is in the appendix.

\begin{lemma}
\label{lem:Gfkoxh}
If an anomalous source $\xko\sim\Gauss{0}{\covko}$ is encoded with the compression scheme $\fenc$ of Lemma~\ref{lem:Gfenc}, then

\begin{equation}
\label{eq:Gfkoxh}
\xhko\sim\Gauss{0}{\Sth\covko\Sth+\theta\Sth}
\end{equation}
\end{lemma}

Such a result has two noteworthy corner cases.
\begin{itemize}
    \item If $\theta\rightarrow 0^+$ there is no distortion. In fact, since $\Sth=I_n$, Lemma~\ref{lem:Gfkoxh} gives $\xhko\sim\xko$.

    \item If $\xko\sim\xok$ there is no anomaly, $\covok=\covko$, and
        \begin{equation*}
        \Sth\covko\Sth+\theta\Sth=[\Sth+\theta(\covok)^{-1}]\covok\Sth=\covok\Sth
        \end{equation*}
        where the last equality holds since $\Sth=\max\left\{0,I_n-\theta(\covok)^{-1}\right\}$, the possible disagreements between $\Sth+\theta(\covok)^{-1}$ and $I_n$ correspond to components multiplied by zero by the last $\Sth$ factor. Hence, Lemma~\ref{lem:Gfkoxh} can be compared with Lemma~\ref{lem:Gfenc} to confirm that $\xhko\sim\xhok$.
\end{itemize}

Lemma~\ref{lem:Gfenc} and Lemma~\ref{lem:Gfkoxh} imply that when the normal and anomalous signals are Gaussian before compression, performance of anomaly detectors depends on how much we are capable of distinguishing between the two distributions in \eqref{eq:Gfokxh} and \eqref{eq:Gfkoxh}.
We quantify the difference between them with two kinds of information-theoretic measures, which model two distinct scenarios, one in which the detector knows both $\fokxh$ and $\fkoxh$ and one in which it knows only $\fokxh$.

To proceed further it is convenient to define the functional
\begin{equation}
L(x';x'') = -\int_{\RR^n} f_{x'}(\alpha) \log_2\left[f_{x''}(\alpha)\right] \dd\alpha    
\end{equation}
that is the average coding rate, measured in bits per symbol, of a source characterized by the PDF $f_{x'}$ with a code optimized for a source with PDF $f_{x''}$, so that $L(x;x)$ is equal to the differential entropy of $x$ \cite[Chapter 8]{Cover_1991}.
As an alternative statistical point of view, if $f_{x'}$ is the PDF of the symbols generated by a source $x'$, $f_{x''}$ is the PDF of the symbols generated by a source $x''$ and $\ell_x(\alpha)=-\log_2 f_{x}(\alpha)$ is the negative log-likelihood that the symbol $\alpha$ has been generated by the sources $x$, then
$L(x';x'')=\Exp\left[\ell_{x''}(\alpha)|x'\right]$, i.e., the average negative likelihood that an instance is generated by the source $x''$ when it is actually generated by the source $x'$.

Within the Gaussian assumption, we can derive the analytical expression for $L$ in the following Lemma whose derivation is in the appendix.

\begin{lemma}
	\label{lem:Gell}
	If $x'\sim\Gauss{0}{\cov'}$ and $x''\sim\Gauss{0}{\cov''}$ then
	
	\begin{equation}
	\label{eq:Gell}
	L(x';x'')=\frac{1}{2\ln 2}\left\{
	\ln\left[(2\pi)^n\left|\cov''\right|\right]
	+\tr\left[(\cov'')^{-1}\cov'\right]
	\right\}
	\end{equation}
	
\noindent where $|\cdot|$ indicates the determinant of its matrix argument.
\end{lemma}

\subsection{ Distinguishability in anomaly-agnostic detection}

When $\fkoxh$ is unknown and only $\fokxh$ is given, we can only consider the average coding rates referring to code optimized for $\xhok$, i.e., $L(\xhko;\xhok)$ and $L(\xhok;\xhok)$. One may quantify the difference between the normal behaviour and an anomalous one by measuring the increase or decrease in the average coding rate with respect to the expected case $L(\xhok;\xhok)$ as follows:
\begin{align}
\zeta &= L\left(\xhko;\xhok\right) - L\left(\xhok;\xhok\right) \label{eq:z-ell}\\
&= \int_{\RR^n} \left[\fokxh(\alpha) - \fkoxh(\alpha)\right] \log_2\fokxh(\alpha) \dd\alpha \label{eq:z-int}
\end{align}

Since there may be anomalies whose encoding yields a lower rate with respect to normal signals, $\zeta$ is not always positive. As a result, a distinguishability measure is given by considering its magnitude, i.e., $|\zeta|$.

From a statistical perspective, $\zeta$ corresponds to the difference in the expectations of the negative log-likelihood that $\alpha$ is normal given either $\alpha$ is actually an instance of $\xhok$ or $\xhko$
\begin{equation*}
    \zeta = \Exp\left[ \ellxhok(\alpha)|\xhko \right] - \Exp \left[ \ellxhok(\alpha)|\xhok \right]
\end{equation*}

The use of the quantity $\ellxhok(\alpha)=-\log_2 \fokxh(\alpha)$ can be found in other anomaly detection related works, e.g., in \cite{Bohm_KDD2009} where it is referred as a {\em coding cost} of $\alpha$.

With the assumption of Gaussian sources, the optimal encoder (in the rate-distortion sense) lets survive only the components $j$ for which $\lambdaok_j > \theta$. Hence, $\fokxh$ and $\fkoxh$ given in \eqref{eq:Gfokxh} and \eqref{eq:Gfkoxh} have only the first $\nth$ components non-null with $n_\theta=\arg\max_j\{\lambdaok_j>\theta\}$. The other $n-\nth$ components are set to $0$ and thus cannot be used to tell anomalous from normal cases. We therefore focus on the first $\nth$ components of $\xhok$ and $\xhko$ which are Gassian with covariance matrices $\covokxh$ and $\covkoxh$ corresponding to the $\nth\times\nth$ upper-left submatrix of $\covok\Sth$ in \eqref{eq:Gfokxh} and of $\Sth\covok\Sth + \theta\Sth$ in \eqref{eq:Gfkoxh}, respectively. 

By properly combining the definition of $\zeta$ in \eqref{eq:z-ell} with the expression of $L$ within the Gaussian assumption in \eqref{eq:Gell}, we obtain
\begin{align}
\label{eq:Gzeta}
\zeta & = \frac{1}{2\ln 2} \tr\left[\covXth-I_{\nth}\right]
\end{align}
where $\covXth = (\covokxh)^{-1}\covkoxh$ which corresponds to the $\nth\times\nth$ upper-left submatrix of $(\covok)^{-1}\covko\Sth + \Tth$. Note that, since $\covXth$ is linear with respect to $\covkoxh$, so is $\zeta$. In addition, $\zeta$ vanishes when $\covokxh=\covkoxh$.

As a noteworthy particular case, when the normal signal is white, i.e., when $\covok=I_n$, we have that $\theta\in[0,1]$ and that for any $\theta<1$, $\Tth=\theta I_n$ and $\nth=n$. Hence, $\covXth=(1-\theta)\covko+\theta I_n$ that leads to $\zeta=0$. This result is not surprising since the distinguishability modelled by $|\zeta|$ depends only on the statistics of $\xok$ that has no exploitable structure.

\subsection{ Distinguishability in anomaly-aware detection}

When both $\fokxh$ and $\fkoxh$ are known, the anomaly detection task reduces to a binary classification problem for which we may resort to the Neyman Pearson Lemma \cite[Theorem 12.7.1]{Cover_1991}, \cite[Theorem 3.1]{Kay_1998}. This lemma can be understood in the sense that the cardinal quantity to observe is 
\begin{equation*}
r(\alpha) = \log_2\left[\frac{\fkoxh(\alpha)}{\fokxh(\alpha)}\right]
\end{equation*}
which can be interpreted as a measure of abnormality of $\alpha$, i.e., a score that the detector employs to distinguish whether the single $\alpha$ behaves normally or not. Consequently, one may measure the distinguishability between the distributions $\fokxh$ and $\fkoxh$ as the difference between the score observed in average when $\xh=\xhok$ and the score obtained in average when $\xh=\xhko$.
\begin{align}
\kappa 
=\;& \Exp\left[ r(\alpha) | \xhko \right]
  - \Exp\left[ r(\alpha) | \xhok \right] \label{eq:kappa-Exp}\\[0.5 em]
\begin{split}
=\;& \int_{\RR^n} \fkoxh(\alpha) \log_2\left[\frac{\fkoxh(\alpha)}{\fokxh(\alpha)}\right] \dd\alpha + \label{eq:kappa-integral}\\
&\;+ \int_{\RR^n} \fokxh(\alpha) \log_2\left[\frac{\fokxh(\alpha)}{\fkoxh(\alpha)}\right] \dd\alpha 
\end{split} \\[0.5 em]
\begin{split}
=\;& L\left(\xhko;\xhok\right) - L\left(\xhko;\xhko\right) + \\
&\;+ L\left(\xhok;\xhko\right) - L\left(\xhok;\xhok\right)
\end{split} \label{eq:kappa-L}\\[0.5 em]
=\;& \DKL{\fkoxh}{\fokxh} + \DKL{\fokxh}{\fkoxh} \label{eq:kappa-DKL}
\end{align}
where, given distributions $f'$ and $f''$, $\DKL{f'}{f''}$ refers to the Kullback-Leibler divergence \cite[Chapter 2]{Cover_1991}, of which $\kappa$ results to be the symmetrized version. 

The measure $\kappa$ models a detector that knows the distributions of both normal and anomalous sources such that their optimal codes are also known. From \eqref{eq:kappa-L}, it is evident that $\kappa$ may be interpreted as the sum of the differences in the average coding rate for both distorted sources with a code optimized for the normal source $L\left(\xhko;\xhok\right) - L\left(\xhok;\xhok\right)$ and optimized for the anomalous source $L\left(\xhok;\xhko\right) - L\left(\xhko;\xhko\right)$. Since the average coding rate is expected to be shorter when employed to code a source for which it is optimized, these differences are expected to be greater when the difference of two distributions $\fokxh$ and $\fkoxh$ increases. As a result, large $\kappa$ values correspond to system configurations with high detection capability.

Differently from $\zeta$, $\kappa$ is a quantity that is always positive and can be directly used as distinguishability measure.

Within the Gaussian assumption, the distiguishability measure $\kappa$ becomes
\begin{align}
\kappa 
&= \frac{1}{2\ln 2}
\tr\left[\covXth
+\covXth^{-1} - 2I_{\nth}\right]
\label{eq:Gkfn}
\end{align}
from which it is evident that $\kappa$ is convex with respect to $\covkoxh$ and, as for $\zeta$, $\kappa$ vanishes for $\covokxh=\covkoxh$.

As a final remark, coherently with the typical {\em per use} analysis, distinguishability measures implicitly consider detectors that scrutiny an increasing number of subsequent source instances and scale their performance by such a number. 
Hence, as rate and distortion coming from \eqref{eq:rate-distortion} are best-case bounds that can be approximated by increasing the complexity of the system, the distinguishability measures indicate how fast a detector accumulates information allowing to declare an anomaly. The higher such a figure, the lower the number of subsequent symbols needed to arrive at a conclusion or, alternatively, the higher the confidence in a conclusion drawn after analyzing as single instance.

\section{Average and large-window distinguishability}
\label{sec:averagelarge}

\subsection{Average on the set of possible anomalies}
\label{subs:possibleanomalies}

Anomalies modelled as zero-mean Gaussian vectors with fixed energy, are completely defined by their covariance matrix $\covko$ where $\tr(\covko)=n$. We decompose $\covko=\Uko\Lambdako{\Uko}^\top$ with $\Lambdako={\rm diag}(\lambdako_0,\dots,\lambdako_{n-1})$ and $\Uko$ orthonormal.

The set of all possible $\lambdako=(\lambdako_0,\dots,\lambdako_{n-1})^\top$ is

\[
\PP^n=\left\{\lambda\in{\RR^+}^n| \sum_{j=0}^{n-1}\lambda_j=n\right\}
\]

\noindent while the set of all possible $\Uko$ is that of orthonormal matrices

\[
\OO^n=\left\{
U\in\RR^{n\times n} | U^\top U=I_n\right\}
\]

By indicating with $\Unif{\cdot}$ the uniform distribution in the argument domain, we will assume that when $\lambdako$ is not known then $\lambdako\sim\Unif{\PP^n}$ and, similarly, when $\Uko$ is not known then $\Uko\sim\Unif{\OO^n}$, independently of $\lambdako$.

Note now that $\PP^n$ is invariant with respect to any permutation of the $\lambda_j$. Since $\lambdako\sim\Unif{\PP^n}$, also $\Exp[\lambdako]$ must be invariant with respect to the same permutations so that $\Exp[\lambdako_j]=\Exp[\lambdako_k]$ for any $j,k$. Since $\lambdaok$ has a constrained sum and is the diagonal of $\Lambdako$ we have $\Exp[\Lambdako]=I_n$.
This implies
\begin{equation}
\label{eq:avecovko}
\begin{split}
\Exp\left[\covko\right] 
&= \Exp\left[\Uko\Lambdako{\Uko}^\top\right] \\ 
&= \Exp\left[\Uko\Exp\left[\Lambdako\right]{\Uko}^\top\right] \\
&= \Exp\left[\Uko {\Uko}^\top\right] = I_n \\
\end{split}
\end{equation}

Hence, in our setting, the average anomaly is white and we may compute the corresponding distinguishability measures $\zeta_I$ and $\kappa_I$, i.e., $\zeta$ and $\kappa$ when $\covko = I_n$. Note that, in this case, $\covXth$ is the $\nth\times\nth$ upper-left submatrix of $(\covok)^{-1}\Sth+\Tth$, which is a diagonal matrix whose diagonal elements are 
\begin{equation*}
\uokthj = 
\frac{1}{\lambdaok_j}\left(1-\frac{\theta}{\lambdaok_j}\right)+
\frac{\theta}{\lambdaok_j}
\end{equation*}

With these quantities, the expressions of the distinguishability measures become

\begin{align}
\label{eq:Gzeta}
\zeta_I & = \frac{1}{2\ln 2}
\sum_{j=0}^{\nth-1} \left(\uokthj-1\right)\\
\label{eq:Gkappa}
\kappa_I & =
\frac{1}{2\ln 2}
\sum_{j=0}^{\nth-1}
\left(
\uokthj + \frac{1}{\uokthj} - 2\right)
\end{align}

Note that due to the Jensen's inequality, the linearity of $\zeta$ and the convexity of $\kappa$, we have $\zeta_I=\Exp[\zeta]$ and $\kappa_I\le\Exp[\kappa]$.

Moreover, the very simple structure of $\zeta_I$ allows the derivation of the following Theorem whose proof is in the appendix.

\begin{theorem}
	\label{th:Zvanishes}
	If $\kb=\arg \max_k\left\{\lambdaok_k\ge \lambdako_k =  1\right\}$, then $\zeta_I=0$ for at least one point $0<\theta<\lambdaok_{\kb}$
\end{theorem}

Considering a white anomaly, the intuition behind this theorem is the following. 
When distortion is null (no compression), since $\xhko$ and $\xhok$ have the same average energy and the coding is tuned on $\xhok$, $L(\xhko, \xhok) > L(\xhok, \xhok)$ such that $\zeta_I$ is positive.
On the other hand, when distortion is so high that only the first component of $\xok$ survives, i.e., $\covokxh=\lambdaok_0-\theta$, a single component also survives in $\xhko$. In this setting, $\zeta_I$ depends on the difference between the two scalar quantities $\covokxh$ and $\covkoxh$. With few numerical manipulations, it is possible to prove that $\covokxh>\covkoxh$ thus $\zeta_I$ results to be negative.
Since $\zeta_I$ is continuous in $\theta$, it must pass through zero at least once. Therefore, at least one critical level of distortion exists that makes the detectors that do not use the information of the anomaly ineffective.

\subsection{Asymptotic distinguishability}

White signals are not only the average anomalies but are also {\em typical} anomalies in a sense specified by the following Theorem whose proof is in the appendix.

\begin{theorem}
	\label{th:covkoconc}
	If $\lambdako\sim\Unif{\PP^n}$ and $\Uko\sim\Unif{\OO^n}$ then, as $n\rightarrow\infty$, $\covko=\Uko{\rm diag}(\lambdako_0,\dots,\lambdako_{n-1}){\Uko}^\top$ tends to $I_n$ in probability.
\end{theorem}

Hence, when $n$ increases, most of the possible anomalies behave as white signals, i.e., $\zeta$ tends to $\zeta_I$, that thus enjoys the property shown in Theorem~\ref{th:Zvanishes}. From an anomaly detection perspective, if the signal is characterized by a sufficiently large dimension $n$, the designer may consider the white anomaly as a reference.

\section{Numerical examples}
\label{sec:numerical}

In this section we match the theoretical derivations with the quantitative assessment of the performance of some practical anomaly detectors applied to compressed signals.

Normal signals are assumed to be $\xok\sim\Gauss{0}{\covok}$ where $\covok$ is the diagonal matrix of the eigendecomposition of the matrix $\cov = U \covok U^\top$ , with $\cov_{j, k}=\omega^{|j-k|}$, for $j, k = 0,\dots,n-1$ and $U$ an orthonormal matrix.
The parameter $\omega$ is set to yield a different degree of non-whiteness measured with the so-called {\em localization} defined as
\begin{equation*}
\Loc_{\xok} 
= \frac{\tr({\covok}^2)}{\tr^2(\covok)}-\frac{1}{n}
\end{equation*}

The localization goes from $\Loc_{\xok}=0$ when the signal is white to $\Loc_{\xok}=1-\frac{1}{n}$ when all the energy is concentrated along a single direction of the signal space (see \cite{Mangia_TCSI2017} for more details). To show the effect of realistic localization \cite{Cambareri_ISCAS2013} we consider values of $\omega$ corresponding to $\Loc_{\xok}\in\{0, 0.05, 0.2\}$.

Anomalous signals are generated as $\xko\sim\Gauss{0}{\covko}$, where $\covko=\Uko\Lambdako{\Uko}^\top$ is randomly picked according to the uniform distribution defined in Section~\ref{subs:possibleanomalies}.

To generate $\lambdako\sim\Unif{\PP^n}$, we follow \cite{Onn_AOR2011} to first draw
$\xi_j\sim\Unif{[0,1]}$ for $j=0,\dots,n-1$ and then set
\[
\lambdako_j=\frac{\log \xi_j}{\sum_{k=0}^{n-1}\log\xi_k}
\]

To generate $\Uko\sim\Unif{\OO^n}$, we follow \cite{Mezzadri_NAMS2011} and start by generating a matrix $A$ within the Ginibre ensemble \cite{Ginibre_JMP1965}, i.e., with independent entries $A_{j,k}\sim\Gauss{0}{1}$ for $j,k=0,\dots,n-1$.
We then set $\Uko$ to the orthonormal factor of the $QR$-decomposition of $A$.

A first use of this random sampling is the possibility of pairing Theorem~\ref{th:covkoconc} with some numerical evidence. Fig.~\ref{fig:covkoconc} reports the vanishing trend of the average squared and uniform deviation from $I_n$ of a population of uniformly distributed covariance matrices $\covko$. Though not a theoretical result, note that empirical evidence supports a classical $\nicefrac{1}{\sqrt{n}}$ convergence. 

\begin{figure}
    \centering
    \includegraphics{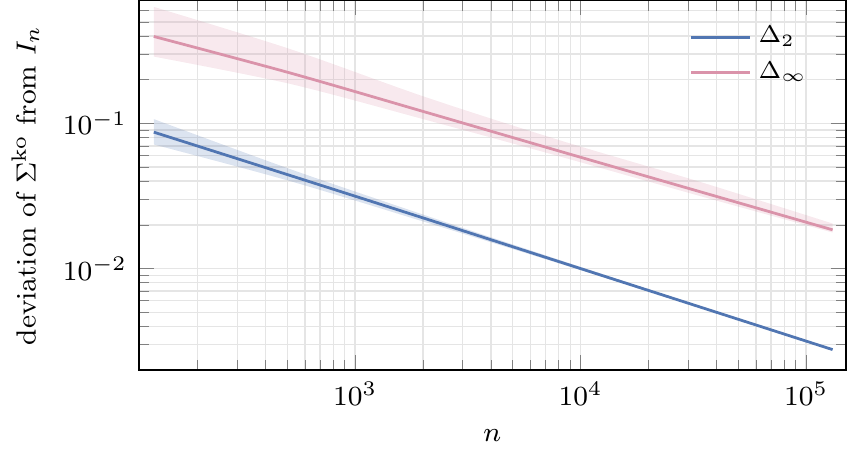}
    \caption{\label{fig:covkoconc}Trend of $\Delta_2=\frac{1}{n}\sqrt{\sum_{j,k=0}^{n-1}\left[\covko_{j,k}-(I_n)_{j,k}\right]^2}$ and $\Delta_\infty=\max_{j,k} \left|\covko_{j,k}-(I_n)_{j,k}\right|$ when $n$ increases. Solid lines are mean trends while shaded areas contain $98\%$ of the population.}
\end{figure}

As far as detector assessment is concerned, we decide to set $n=32$ and consider three compression techniques tuned to the normal signal and applied to both normal and anomalous instances. More specifically, $x$ is mapped to $\xh$ by

\begin{itemize}
    \item the minimum-rate-given-distortion compression in \eqref{eq:Gfenc} (Rate-Distortion Compression {\tt RDC});
    \item projecting $x$  along the subspace spanned by the eigenvectors of $\covok$ with the largest eigenvalues (Principal Component Compression {\tt PCC}); 
    \item a family of autoencoders \cite[Chapter 14]{Goodfellow_2016} with an increasingly deficient latent representation (Auto-Encoder Compression {\tt AEC}).
    Assuming that $p$ is the dimensionality of the representation, the encoder is a neural network with fully connected layers of dimensions $n$, $4n$, $2n$, $p$, and the decoder is the dual network whose layers have dimensions $2n$, $4n$, and $n$ and the number of inputs is $p$. The family of autoencoders is trained to minimize distortion computed as in \eqref{eq:distortiondef}. To smooth performance degradation we first train an autoencoder with $p=n-1$.
    Then, the node of the latent representation along which we measure the least average energy is dropped to produce a smaller network with an $(p-1)$-dimensional latent space. The obtained network is re-trained using the previous weights as initialization. This process is repeated decreasing $p$ and thus considering larger distortion values.
\end{itemize}

These three schemes address in a different way the trade-off between compression and distortion.
Since we refer to a theoretical model based on continuous quantities and for which rate is potentially infinite, the compressors have to be  paired with a quantization stage ensuring that rate values are finite. In particular we encode each component of $\xh$ with $16$ bits and this yields rates of less than $16n=512$ bits per time step. We assume that quantization is fine enough to substantially preserve the Gaussian distribution of $\xh$ and thus evaluate the mutual information between $x$ and $\xh$ as if they were jointly gaussian with a covariance matrix that we estimate by Monte Carlo simulation \cite{Arellano_SJS2013}.
Such estimation yields the rate-distortion curves in Fig.~\ref{fig:rate_distortion}.

\begin{figure}
    \centering
    \includegraphics[width=\columnwidth]{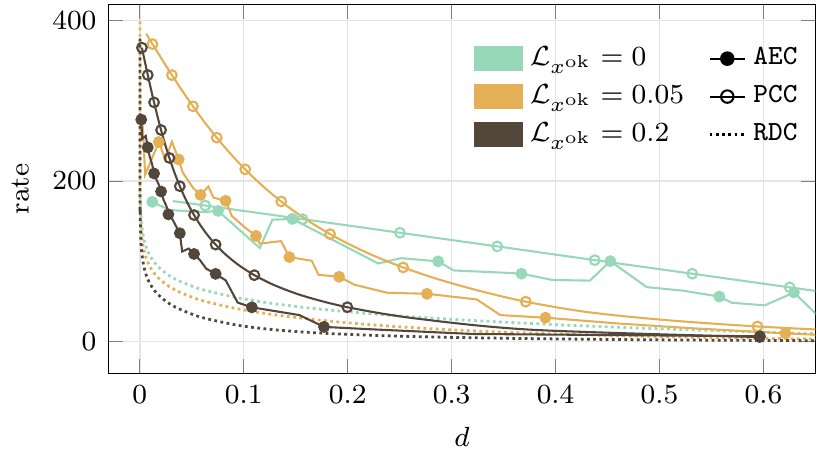}
    \caption{\label{fig:rate_distortion}Rate distortion curves for the three compression schemes we consider and for different value of the localization of the original signal.}
\end{figure}

As expected, {\tt RDC} yields the smallest rates while {\tt PCC} gives the largest ones. Between the two we have {\tt AEC}, whose performance depends on the effectiveness of the training.

Note that, only the results of Fig.~\ref{fig:rate_distortion} refer to the additional quantization stage, while in the remaining part of our analysis we consider continuous sources.

The compressed version of the signal is then passed to a detector whose task is to compute a score such that high-score instances should be more likely to be anomalous. The final binary decision is taken by matching the score against a threshold.

We consider two detectors not relying on information of the anomaly
\begin{itemize}
    \item a Likelihood Detector ({\tt LD}) whose score is the same considered for $\zeta$, so that to each instance $x$ we associate the score $\ellxhok(\xh) = -\log \fokxh(\xh)$;
    \item a One-Class Support-Vector Machine ({\tt OCSVM}) \cite{Scholkopf_NIPS1999} with a Gaussian kernel\footnote{The signal components are normalized by their variance and the scale parameter of the Gaussian kernel is fixed to $1/n_\theta$.}, trained on a set of instances of normal signals contaminated by 1\% of unlabelled white instances to help the algorithm in finding the envelope of normal instances.
\end{itemize}

We also consider two detectors that are able to leverage information on the anomaly
\begin{itemize}
\item a Neyman-Pearson Detector ({\tt NPD}), whose score is the same considered for $\kappa$, so that to instance $x$ we associate the score $r(\xh) = \log\fkoxh(\xh)-\log\fokxh(\xh)$;
\item a Deep Neural Network ({\tt DNN}) with three fully connected hidden layers with $p$, $2n$, $n$ neurons with ReLu activations and a final sigmoid neuron producing the score. The network is trained\footnote{Training of {\tt DNN} involves the backpropagation algorithm with ADAM optimizer \cite{Kingma_ICLR2015}, a batch-size of $20$ instances, and an initial learning rate of $0.01$ that is scaled by $0.2$ any $5$ epochs for which the validation loss is not decreasing where the validation set consists of the 10\% of the instances initially devoted to the training. These training parameters are the result of a tuning.} with a binary cross-entropy loss against a dataset containing labelled normal and anomalous instances.
\end{itemize}

{\tt LD} and {\tt NPD} detectors can be employed only on signals compressed by {\tt RDC} or by {\tt PCC} method since they rely on the statistical characterization of the signals that is not available after the nonlinear processing in {\tt AEC}.

\begin{table}
    \centering
    \caption{\label{tab:samples}Number of anomalies ($\covko$) and, for each anomaly, the number of normal ($\ok$) and anomalous ($\ko$) signal instances used in the training and assessment of the detectors.}
    \begin{tabular}{rrrrrrrr}
    \toprule
        \multirow{ 3}{*}{detector} &
        \multicolumn{3}{c}{training} & &
        \multicolumn{3}{c}{assessment}\\
        \cmidrule{2-4}
        \cmidrule{6-8}
                   &
        \multicolumn{1}{c}{\#$\covko$} & \multicolumn{2}{c}{\#instances$\times\covko$} & &
        \multicolumn{1}{c}{\#$\covko$} & \multicolumn{2}{c}{\#instances$\times\covko$} \\
        & & \multicolumn{1}{c}{$\ok$} & \multicolumn{1}{c}{$\ko$} & 
        & & \multicolumn{1}{c}{$\ok$} & \multicolumn{1}{c}{$\ko$}\\
    \midrule
         LD &
         & & & &
         $10^3$ & $10^3$ & $10^3$\\
         OCSVM &
         1 & $99\!\times\!\!10^3$ & $10^3$ & & $10^3$ & $10^3$ & $10^3$\\
         NPD & 
         & & & &
         $10^3$ & $10^3$ & $10^3$\\
         DNN & 
         50 & $10^5$ & $10^5$ & &
         50 & $10^5$ & $10^5$\\
    \bottomrule
    \end{tabular}
\end{table}

\begin{figure*}
    \centering
    \includegraphics{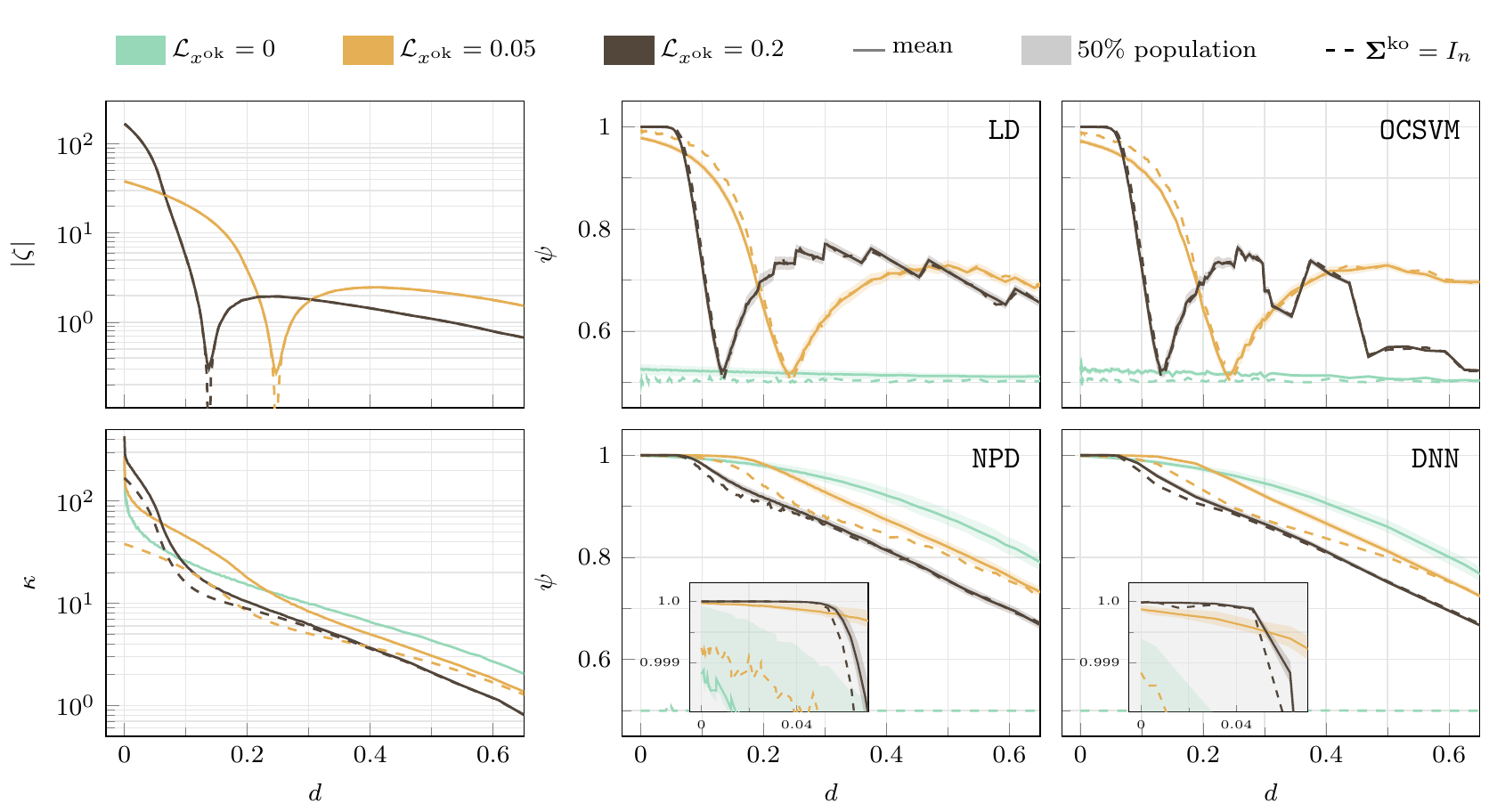}
    \caption{\label{fig:opt_compressor} Distinguishability measures $\zeta$, $\kappa$ and $\psi$ against normalized distortion $d$ in case of {\tt RDC}.}
\end{figure*}
\begin{figure*}
    \centering
    \includegraphics{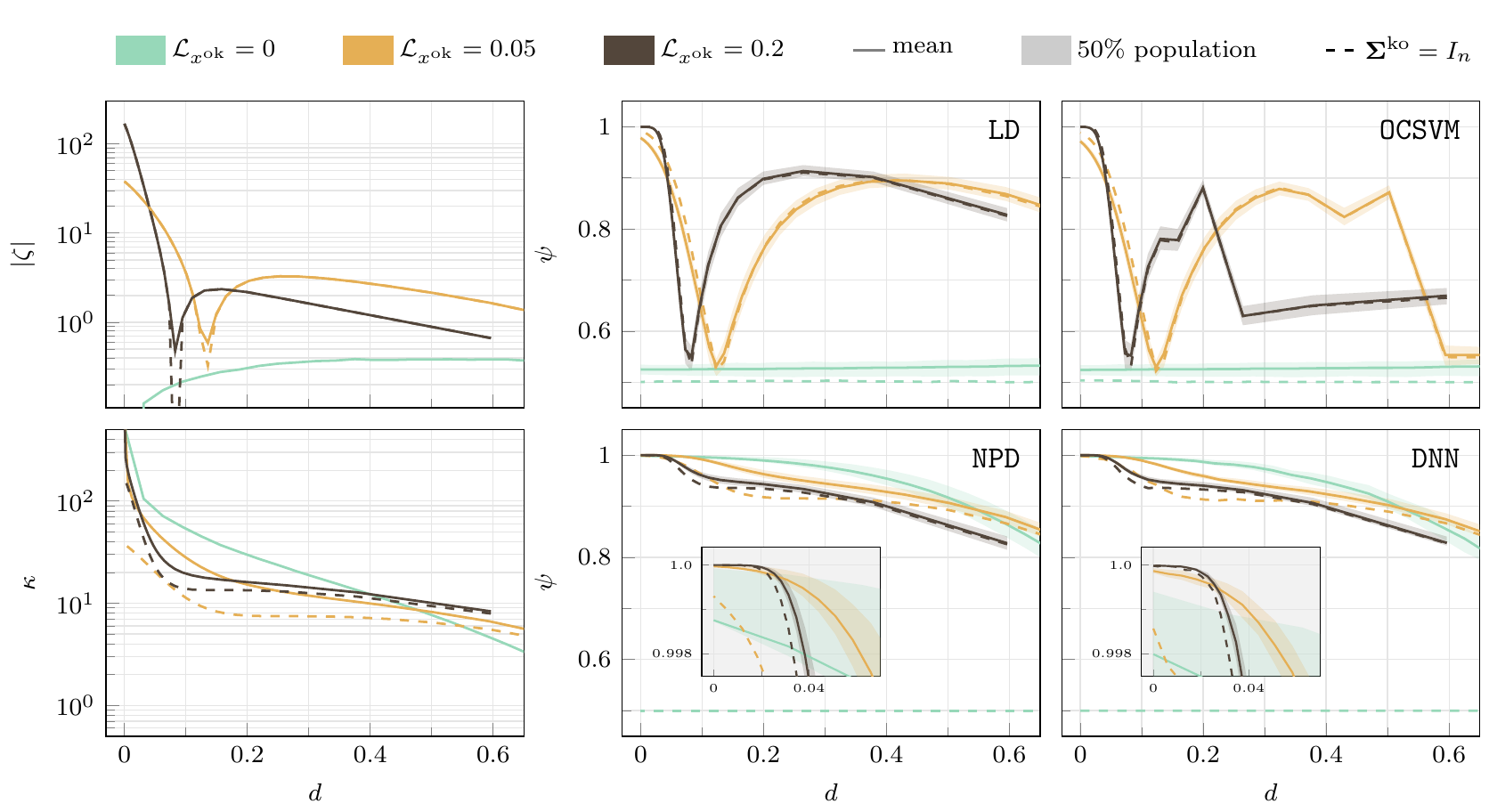}
    \caption{\label{fig:pca_compressor} Distinguishability measures $\zeta$, $\kappa$ and $\psi$ against normalized distortion $d$ in case of {\tt PCC}.}
\end{figure*}

Table \ref{tab:samples} shows how many different anomalies and how many signal instances are generated for the training (when needed) and for the assessment of the detectors. Note that in the DNN case we limited the analysis to $50$ anomalies since the training process must be repeated for each of them.

To be independent of the choice of thresholds, detectors' performance is assessed by the Area-Under-the-Curve ($\AUC$) methodology \cite{Fawcett_PATREC2006}. $\AUC$ estimates the probability that given a random normal instance and a random anomalous instance, the former has a lower score with respect to the latter, as it should be in an ideal setting. Hence, $\AUC$ is a positive performance index.

Clearly, detectors with $\AUC=\frac{1}{2}$ are no better than coin tossing. Yet, if $\AUC<\frac{1}{2}$, the score has some ability to distinguishing normal and anomalous signals if it is interpreted in a reverse way. Hence, it is convenient to set our {\em empirical distinguishability} measure to
\begin{equation*}
\psi = \begin{cases}
\AUC & \text{if $\AUC \ge \frac{1}{2}$}\\
1 - \AUC & \text{if $\AUC < \frac{1}{2}$}\\
\end{cases}
\end{equation*}

Note that, if $\AUC$ must be estimated from samples, reversing values lower than $\frac{1}{2}$ is not always possible. There are classes of estimators for which values less than $\frac{1}{2}$ are not reliable \cite{Jamalabadi_HBM2016, Snoek_NeuroImage2019}. From now on, we report results referring to $\AUC$ estimated as in \cite{Fawcett_PATREC2006} for which reversing values lower than $\frac{1}{2}$ is possible.

In the following, the trends of $\psi$ are reported and matched with the trends of $|\zeta|$ and $\kappa$ to show how theoretical properties reflect on real cases.
Comparisons must be partially qualitative as $\zeta$ and $\kappa$ quantify the distinguishability with bits per symbol while $\psi$ comes from the probability of correct detection. Note also that $\zeta$ and $\kappa$ refer to the difference between the average values of the score in the normal and anomalous cases, while $\psi$ takes into account the entire distributions of these scores.

All plots are made against a normalized distortion $d = D/n$ in the range $d\in\left[0,0.64\right]$ as larger relative distortions are usually beyond operative ranges.

\subsection{\tt RDC}

Fig.~\ref{fig:opt_compressor} summarizes the results we have in this case with two rows of 3 plots each.
The upper row of plots corresponds to detectors that do not exploit information on the anomaly, while the lower row of plots concerns detectors that may leverage information on the anomaly.
Colors correspond to different $\Loc_{\xok}$, dashed trends assume that the anomaly is the average one, i.e., white, and shaded areas show the span of $50\%$ of the Monte Carlo population. 
The profiles of $|\zeta|$ and $\kappa$ on the left shall be matched with the profiles on the right that correspond to the four detectors we consider. No $|\zeta|$ profile appears for $\Loc_x=0$ as in that case $\zeta=0$.

Numerical results confirm that $\zeta_I = \Exp[\zeta]$, $\kappa_I \le \Exp[\kappa]$ as discussed in Section~\ref{subs:possibleanomalies}. This corroborates the role of the white anomaly as a reference case since it represents the average behaviour in case of anomaly-agnostic detector or a lower bound in the anomaly-aware scenarios. The white anomaly is not only a reference case but also the case to which any possible anomaly tends when $n$ increases as demonstrated in Theorem~\ref{th:covkoconc}.

Theory also anticipates that without any knowledge of the anomaly (upper row), a limited amount of distortion may cause distinguishability to vanish and thus detectors to fail. This happens for practical detectors such as LD and OCSVM. The distortion level at which detectors fail is also anticipated by $|\zeta|$ and depends on $\Loc_{\xok}$ as predicted by  Theorem~\ref{th:Zvanishes}.
Overall, theoretical measures $|\zeta|$ and $\kappa$ anticipate that in the low-distortion region, more localized signals are more distinguishable from anomaly though they cause detector failures at smaller distortions with respect to less localized signals.

Detectors leveraging the knowledge of the anomaly (lower row) fail completely only at the maximum level of distortion as revealed by the abstract distinguishability measure $\kappa$. 
Also in this case, by comparing the trend of $\kappa$ with the zoomed areas in the NPD and DNN plots we see how theoretical measures anticipates that in the low-distortion region more localized signals tend to be more distinguishable from anomalies but cause a more definite performance degradation of detectors when $d$ increases.

\subsection{\tt PCC}

From the point of view of the rate-distortion trade-off {\tt PCC} is largely suboptimal. Yet, due to its linear nature, $x$ and $\xh$ are still jointly Gaussian, so that, also in this case, we can compute the theoretical $|\zeta|$ and $\kappa$ by means of \eqref{eq:Gzeta} and \eqref{eq:Gkfn}.

Fig.~\ref{fig:pca_compressor} summarizes the results we have in this case with plots of the same kind of Fig.~\ref{fig:opt_compressor}. The qualitative behaviours commented in the previous subsection appear in the new plots and are anticipated by the trends of the theoretical quantities.

The distortion levels at which anomaly-agnostic detectors fail change with respect to the {\tt RDC} case but are still anticipated by the theoretical curves and Theorem~\ref{th:Zvanishes}.

In this case, the values of $|\zeta|$ beyond breakdown distortion levels increase slightly more that in the optimal compression scenario. Hence, by adopting a compression strategy that is suboptimal in the rate-distortion sense one may obtain a better distinguishability of the compressed normal signal from the compressed anomalies.
This is, indeed, what happens in practice as highlighted by the LD and OCSVM plots in the first row of Fig.~\ref{fig:pca_compressor}.

\subsection{\tt AEC}

In this case, compression is non-linear so that $x$ and $\xh$ may not be jointly Gaussian. This prevents us from computing the theoretical curves $|\zeta|$ and $\kappa$ and from applying {\tt LD} and {\tt NPD} that rely on the knowledge of the distribution of the signals. For this reason, Fig.~\ref{fig:ae_compressor} reports only the performance of OCSVM and DNN detectors.

Notice how the qualitative trends of those performances still follow, though with a larger level of approximation, what is indicated by the theoretical curves for {\tt PCC}. 

\begin{figure}
    \centering
    \includegraphics{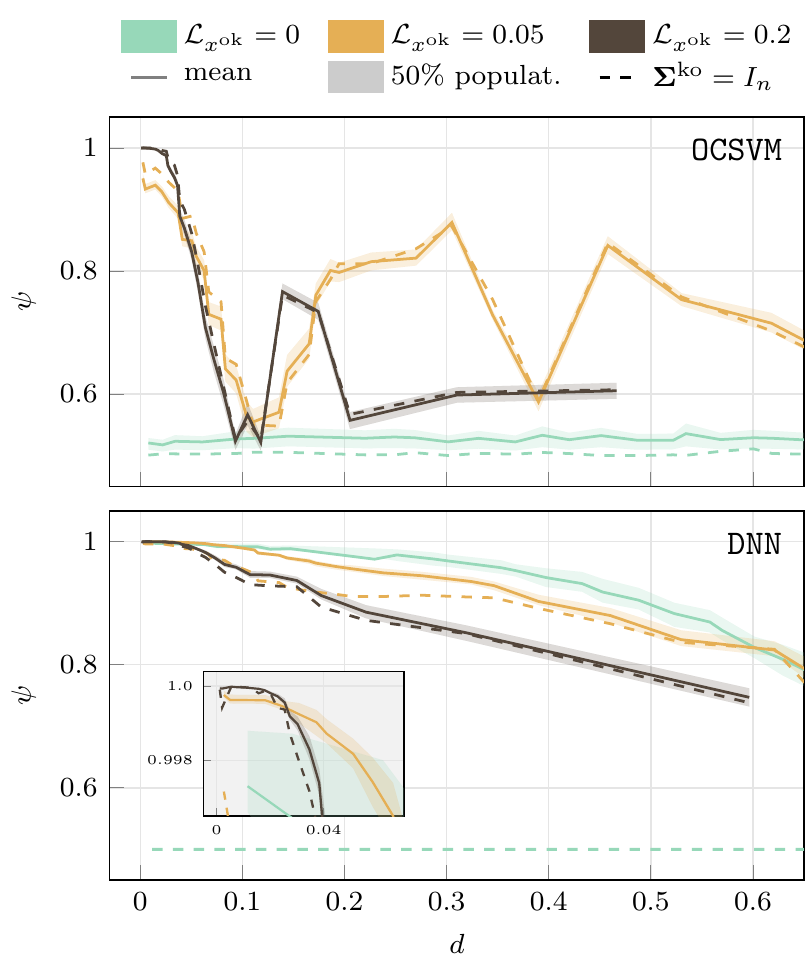}
    \caption{\label{fig:ae_compressor} Distinguishability measure $\psi$ against normalized distortion $d$ in case of {\tt AEC}.}
\end{figure}

\section{Conclusion}

Massive sensing systems may rely on lossy compression to reduce the bitrate needed to transmit acquisitions to the cloud while theoretically maintaining the important information.
At some intermediate point along their path to centralized servers, compressed sensor readings may be processed for early detection of anomalies in the systems under observation.
Such detection must be performed on compressed data.

To measure detection performance we define two information theory metrics referring to the anomaly-agnostic and anomaly-aware cases, for which a statistical interpretation is also provided.

In a framework approximating normal and anomalous signals with Gaussian sources, we revise the classical rate-distortion theory to report the distributions of the distorted signals, the mapping to obtain them (see Lemma~\ref{lem:Gfenc} and Lemma~\ref{lem:Gfkoxh}), and closed forms for the distinguishability metrics.

Focusing on the anomaly-agnostic case, we prove with Theorem~\ref{th:Zvanishes}, and confirm with numerical evidences, that there exists at least one critical level of distortion for which the detector is ineffective.

We also prove that the white anomaly is a reference case that can be employed in the design of the system. Indeed, it provides information about the average and minimum performance in the anomaly-agnostic and anomaly-aware scenarios, respectively. Moreover, we demonstrate with Theorem~\ref{th:covkoconc} that any possible anomaly tends to be white in the asymptotic case. 

All these results are confirmed with numerical examples in a toy case. We show that the theoretical measure of distinguishability anticipate the performances of real detectors in case of both optimal (in the rate-distortion sense) and sub-optimal compressors.
 
\section*{Appendix}

\begin{proof}[Proof of Lemma \ref{lem:Gfenc}]

Distortion is tuned to the normal case that entails a memoryless sources. Hence we may drop time indications and concentrate on a vector $x$ with independent components $x_j\sim\Gauss{0}{\lambda_j}$ for $j=0,\dots,n-1$.

We know from \cite{Kolmogorov_TIT1956} that for a given value of the parameter $\theta$, each component $x_j$ is transformed separately into $\xh_j$. In particular,
\begin{equation}
    \xh_j = \begin{cases} 
    0 & \text{if $\lambda_j\le\theta$} \\
    x_j+\Delta_j & \text{if $\lambda_j>\theta$}
    \end{cases}
\end{equation}
where, to achieve the Shannon lower bound, $\Delta_j$ must be an instance of a Gaussian random variable independent of $\xh_j$. Hence, the three quantities $\xh_j$, $x_j$ and $\Delta_j$ must be such that $\left(\xh_j,x,\Delta_j\right)^\top\sim\Gauss{0}{\cov_{\xh_j,x_j,\Delta_j}}$ with

\begin{equation}
\label{eq:covall}
\cov_{\xh_j,x_j,\Delta_j}=
\begin{pmatrix}
\lambda_j-\theta & \lambda_j-\theta & 0 \\
\lambda_j-\theta & \lambda_j & -\theta\\
0 & -\theta & \theta
\end{pmatrix}
\end{equation}

That explains in which sense $\xh_j$ {\em encodes} $x_j$. In fact, the non-diagonal elements $\lambda_j-\theta$ are positive and thus $\xh_j$ and $x_j$ are positively correlated.

From \eqref{eq:covall}, if we agree to identify a Gaussian with $0$ variance with a Dirac's delta we infer that $\xh_j\sim\Gauss{0}{\max\{0,\lambda_j-\theta\}}$ and thus $\xh\sim\Gauss{0}{\cov\Sth}$. 

Moreover, $\left(\xh_j, x_j\right)^\top \sim \Gauss{0}{\cov_{\xh_j,x_j}}$ with $\cov_{\xh_j, x_j}$ the upper-left $2\times 2$ submatrix of $\cov_{\xh_j,x_j,\Delta_j}$ in \eqref{eq:covall}. If we assume that $\theta<\lambda_j$, from the joint probability of $x_j$ and $\xh_j$, we may compute the action of $\fxhgx$ on the $j$-th component of $x_j$ as the PDF of $\xh_j$ given $x_j$, i.e.,
\begin{align*}
f_{\xh_j|x_j}(\alpha,\beta) 
&= \frac{f_{\xh_j,x_j}(\alpha,\beta)}{f_{x_j}(\beta)}
= \frac{
\GPDF{\begin{matrix}\alpha\\\beta\end{matrix}}{0}{\cov_{\xh_j,x_j}}
}{
\GPDF{\beta}{0}{\lambda_j}
}\\
&= \frac{1}{\sqrt{2\pi\lambda_j\tau_j s_j}}
\exp\left({-\frac{1}{2}
\frac{\left[\alpha - s_j\beta\right]^2}{\lambda_j\tau_j s_j}}\right)
\end{align*}
where  $\tau_j=\min\{1,\theta/\lambda_j\} \in [0,1]$, and $s_j=1-\tau_j$. Note that, $f_{\xh_j|x_j}$ becomes $\delta(\alpha)$ for $\tau_j\rightarrow 1$ (maximum distortion of this component implies that the corresponding output is set to $0$) and $\delta(\alpha-\beta)$ for $\tau_j\rightarrow 0$ (no distortion of this component, the output is equal to the input).

We may collect the component-wise PDFs into a vector PDF by using the matrix $\Tth={\rm diag}\left(\tau_0,\dots,\tau_{n-1}\right)=\min\{I_n,\theta(\covok)^{-1}\}$, and the matrix $\Sth=I_n-\Tth$ thus yielding the thesis.
\end{proof}

\begin{proof}[Proof of Lemma \ref{lem:Gfkoxh}]

The PDF of $\xhko$ distorted by means of $\fenc$ can be computed as

\[
\fkoxh(\alpha)=\int_{\RR^n}\fkoxhjx(\alpha,\beta)\dd\beta=\int_{\RR^n}\fenc(\alpha,\beta)\fkox(\beta)\dd\beta
\]

Assume first to be in the low-distortion condition $\theta<\lambdaok_{n-1}$ that implies $\Tth=\theta(\covok)^{-1}$, and write

\begin{eqnarray*}
\lefteqn{\fkoxh(\alpha)=
\int_{\RR^n}
\GPDF{\alpha}{\Sth\beta}{\covok\Sth\Tth}
\GPDF{\beta}{0}{\covko}\dd \beta}\\
&=&
\GPDF{\alpha}{0}{\covok\Sth\Tth}\times\\
&&
\int_{\RR^n}
e^{-\frac{1}{2}\left[\beta^\top\Sth (\covok\Sth\Tth)^{-1}\Sth\beta-2
\alpha(\covok\Sth\Tth)^{-1}\Sth\beta\right]}\times\\
&&\hspace{50mm}
\GPDF{\beta}{0}{\covko}\dd \beta
\\
&=&
\GPDF{\alpha}{0}{\covok\Sth\Tth}\times\\
&&
\frac{1}{\sqrt{(2\pi)^{n} \det \covko}}
\underbrace{
	\int_{\RR^n}
	e^{-\frac{1}{2}\left(
		\beta^\top Q \beta-
		2 q^\top \beta
		\right)
	}\dd\beta}_{g(\alpha)}\hspace{5mm}
\end{eqnarray*}

\noindent with $Q=\Sth(\covok\Sth\Tth)^{-1}\Sth+(\covko)^{-1}=(\theta I_n)^{-1} -(\covok)^{-1}+(\covko)^{-1}$ and $q=(\covok\Sth\Tth)^{-1}\Sth\alpha=\alpha/\theta$.
To compute $g(\alpha)$ let $Q=UDU^\top$ with $D$ diagonal and $U$ orthonormal, and
set $\beta'=D^{\nicefrac{1}{2}}U^\top \beta$ so that $\beta=UD^{-\nicefrac{1}{2}}\beta'$ and $\dd \beta=\nicefrac{\dd \beta'}{\sqrt{\det Q}}$.
With this write

\[
g(\alpha)=
\frac{1}{\sqrt{\det Q}}
 \int_{\RR^n}e^{-\frac{1}{2}\left(\beta'^\top \beta'-2q^\top U D^{-\nicefrac{1}{2}}\beta'\right)}\dd\beta'
\]

\noindent at the exponent of which one may add and subtract $q^\top Q^{-1} q=q^\top UD^{-\nicefrac{1}{2}} D^{-\nicefrac{1}{2}}U^\top q$ to yield

\begin{align*}
g(\alpha)&=
\frac{1}{\sqrt{\det Q}}
\int_{\RR^n}e^{-\frac{1}{2}
\left(
\left\|	
\beta'-	
D^{-\nicefrac{1}{2}}U^\top q
\right\|^2
-q^\top Q^{-1} q
\right)}\dd\beta'\\
&=
\sqrt{\frac{(2\pi)^n}{\det Q}}
e^{\frac{1}{2}q^\top Q^{-1} q}
\end{align*}

Putting this back into $\fokxh$ we get

\[
\fkoxh(\alpha)=
\GPDF{\alpha}{0}
{
\left[(\theta I_n)^{-1} -(\covok)^{-1}+(\covko)^{-1}\right]\covko\covok\Sth\Tth
}
\]

A straightforward expansion of the definitions under the low-distortion assumption finally rearranges the covariance matrix into 

\begin{eqnarray}
\nonumber
\lefteqn{\left[(\theta I)^{-1} -(\covok)^{-1}+(\covko)^{-1}\right]\covko\covok\Sth\Tth=}\\
\nonumber
&=&
\left[(\theta I_n)^{-1} -(\covok)^{-1}+(\covko)^{-1}\right]\covko\covok\theta(\covok)^{-1}\Sth
\\
\nonumber
&=&
\left[\covko-\theta(\covok)^{-1}\covko+\theta I_n\right]\Sth\\
\nonumber
&=&
\left[I_n-\theta(\covok)^{-1}\right]\covko\Sth+\theta\Sth\\
\label{eq:covkoxh}
&=&\Sth\covko\Sth+\theta\Sth
\end{eqnarray}

\noindent as in the statement of the Lemma.

To address the case in which $\theta$ exceeds $\lambdaok_{n-1}$ note that for 
$\theta\rightarrow (\lambdaok_{n-1})^-$, the last diagonal entry of $\Sth$ tends to $0$ and thus by \eqref{eq:covkoxh} the covariance tends to have zeros in its last row and column. Since a Gaussian with vanishing-variance can be considered Dirac's delta, this model the fact that the last component of both $x$ and $\xko$ is fully distorted and set to $0$.
With this, \eqref{eq:covkoxh}, is valid also for $\lambdaok_{n-1}<\theta<\lambdaok_{n-2}$.
Yet, analogous considerations can be carried out for $\theta\rightarrow(\lambdaok_j)^-$ and $j=n-2,n-3,\dots,0$ so that \eqref{eq:covkoxh} is valid for any value of $\theta$.
\end{proof}

\begin{proof}[Proof of Lemma \ref{lem:Gell}]

\begin{eqnarray*}
\lefteqn{L(x';x'')= -\int_{\RR^n}\GPDF{\alpha}{0}{\cov'}\log_2\left[\GPDF{\alpha}{0}{\cov''}\right]\dd\alpha}\\
&=&
\frac{1}{2}\log_2\left[(2\pi)^n\left| \cov''\right|\right]\int_{\RR^n}\GPDF{\alpha}{0}{\cov'}\dd\alpha\\
&&
\hspace{15mm}+\frac{1}{2\ln 2}\int_{\RR^n}\alpha^\top(\cov'')^{-1}\alpha\,\, \GPDF{\alpha}{0}{\cov'}\dd\alpha\\
&=&
\frac{1}{2}\log_2\left[(2\pi)^n\left|\cov''\right|\right]
+\frac{1}{2\ln 2}\tr\left[(\cov'')^{-1}\cov'\right]
\end{eqnarray*}

Where the last summand has been computed as the expectation of a quadratic form in a Gaussian multivariate for which Corollary 3.2b.1 in \cite[chapter 3]{Provost_1992} gives a formula.	

\end{proof}

\begin{proof}[Proof of Theorem \ref{th:Zvanishes}]

From \eqref{eq:Gzeta} we have that
\begin{equation*}
\zeta_I = \frac{1}{2\ln 2}\sum_{j=0}^{\nth-1}\alpha_j(\theta)
\end{equation*}
with
\begin{equation*}
\alpha_j(\theta) = \frac{1}{\lambdaok_j} \left(1 - \frac{\theta}{\lambdaok_j}\right) + \frac{\theta}{\lambdaok_j} - 1
\end{equation*}

Note that $\alpha_j(\theta)$ is continuous and its derivative is $\frac{\partial}{\partial \theta}\alpha_j = (1-\nicefrac{1}{\lambdaok_j})/\lambdaok_j$.

For simplicity's sake assume $\lambdaok_0>\lambdaok_1>\dots>\lambdaok_{n-1}>0$, set $\lambdaok_n=0$, and define $\Theta_j=]\lambdaok_{j+1},\lambdaok_{j}[$ for $j=0,\dots,n-1$ so that if $\theta\in\Theta_j$ then $\nth=j+1$.

As a function of $\theta$, $\zeta_I$ is continuous.
In fact, it is trivially continuous in each $\Theta_j$. Yet, it is continuous also at any chosen  $\lambdaok_{\jb}$ with $\jb=0,\dots,n-1$. To see why, note that
\begin{align*}
\lim_{\theta\rightarrow{\lambdaok_{\jb}}^-} \zeta_I
& =
\frac{1}{2\ln 2} \lim_{\theta\rightarrow{\lambdaok_{\jb}}^-}
\sum_{j=0}^{\jb}\alpha_j(\theta)\\
& =
\frac{1}{2\ln 2} \lim_{\theta\rightarrow{\lambdaok_{\jb}}^-}
\alpha_{\jb}(\theta)+
\sum_{j=0}^{\jb-1}\alpha_j(\theta)\\
& =
\frac{1}{2\ln 2} \lim_{\theta\rightarrow{\lambdaok_{\jb}}^+}\sum_{j=0}^{\jb-1}\alpha_j(\theta)
=
\lim_{\theta\rightarrow{\lambdaok_{\jb}}^+} \zeta_I
\end{align*}
where we have exploited that the $\alpha_j(\theta)$ are continuous and thus their left and right limits coincide, and that $\alpha_{\jb}(\lambdaok_{\jb})=0$.

On the left-hand side of its domain, When $\theta=\lambdaok_n=0$ (no distortion), we have $\nth=n$ and thus
\begin{equation*}
\zeta_I = \frac{1}{2\ln 2} \sum_{j=0}^{n-1} \left(\frac{1}{\lambdaok_j} - 1 \right) \ge 0
\end{equation*}
\noindent where the last inequality follows from the fact that $\sum_{j=0}^n\lambdaok_j=n$ and thus $\sum_{j=0}^n\nicefrac{1}{\lambdaok_j}\ge n$.

On the right-hand side of its domain, when $\theta=\lambdaok_0$ (maximum distortion), we have $\nth=0$ and thus $\zeta_I=0$.
Yet, we also have that
\begin{equation*}
\frac{\partial}{\partial\theta} \zeta_I = \frac{1}{2\ln 2}
\sum_{j=0}^{\nth-1}
\frac{1}{\lambdaok_j} \left(1 - \frac{1}{\lambdaok_j}\right)
\end{equation*}
\noindent in which the summands are positive if $\lambdaok_j>1$. Hence, if $\kb=\arg\max_k\{\lambdaok_k\ge 1\}$, for $\theta\ge\lambdaok_{\kb}$, all the summands in the above expression  are positive and 
thus $\frac{\partial}{\partial\theta}\zeta_I > 0$ for $\lambdaok_{\kb}<\theta \le \lambdaok_0$. 
Given that $\zeta_I=0$ at the end of that interval, it must be negative in its interior.

Since we know that $\zeta_I$ is positive for $\theta=\lambdaok_n=0$ and it is continuous for $\theta\in]\lambdaok_n,\lambdaok_0[$, it must pass through zero at least once whenever it is not negative, i.e., for $0<\theta<\lambdaok_{\kb}$.
\end{proof}

\begin{proof}[Proof of Theorem~\ref{th:covkoconc}]

We will use the following Lemma whose proof follows this one.

\begin{lemma}
\label{lem:lambdaExp}
If $\lambdako\sim\Unif{\PP^n}$, then
for any integrable function $f:\RR\rightarrow\RR$ and any $j=0,\dots,n-1$

\[
\Exp\left[f(\lambdako_j)\right]=\frac{n-1}{n^{n-1}}\int_0^n f(p)(n-p)^{n-2}\dd p
\]
\end{lemma}

From \cite{Jiang_AOP2006} we know that if $Q$ is uniformly distributed in $\OO^n$ (i.e., if it is distributed according to the Haar measure on the orthogonal group) then, for any sequence of integers $M_n<n$ increasing with $n$ but such that $M_n=o\left(\frac{n}{\log n}\right)$, the entries of the first $M_n$ columns of $Q$ converge in probability to independent random variables such that $\sqrt{n}Q_{j,k}\sim\Gauss{0}{1}$ for $j=0,\dots,n-1$, $k=0,\dots,M_n-1$.

Such a property can be extended to any subset of $M_n$ columns. In fact, given any subset of $M_n$ columns of $Q$, there is a permutation matrix $P$ such that $QP$ has such columns as the first ones.

Yet, since $P\in\OO^n$ and $Q$ is distributed according to the Haar measure in that group, also $QP$ is distributed according to that measure and the entries in those columns tend to independent Gaussians.

Divide now $n$ by $M_n$ as in $n=N_n M_n+m_n$, where $N_n$ is the quotient and $0\le m_n<M_n$ is the remainder. We can look at $Q$ as the concatenation of $N_n$ matrices $Q_i$, for $i=0,\dots,N_n-1$, each $n\times M_n$, and of a last matrix $Q_{N_n}$ that is $n\times m_n$. 
From $M_n=o\left(\frac{n}{\log n}\right)$, we have $\lim_{n\rightarrow \infty} N_n/\log n = \infty$ and we can choose $M_n$ such that $N_n=o(n^\alpha )$ for any $\alpha>0$.

If we set $\Uko=Q$ then, $\covko=\Uko\Lambdako{\Uko}^\top$ can be written componentwise

\begin{eqnarray*}
\lefteqn{\covok_{j,k}=\sum_{l=0}^{n-1}\Uko_{j,l}\lambdako_l\Uko_{k,l}=}\\
&=&
\sum_{i=0}^{N_n-1}
\underbrace{
\sum_{l=iM_n}^{(i+1)M_n-1}\frac{1}{n}\nu_{j,l}\lambdako_l \nu_{k,l}}_{W_i}+
\underbrace{
\sum_{l=n-m_n}^{n-1}\frac{1}{n}\nu_{j,l}\lambdako_l \nu_{k,l}}_{W_{N_n}}
\end{eqnarray*}

\noindent $i=0,\dots,N_n-1$ scans the first submatrices $Q_i$, the last summand accounts for the {\em remainder} matrix $Q_{N_n}$, and, thanks to the above considerations, $\nu_{j,l}\sim \nu_{k,l}\sim\Gauss{0}{1}$ for all $j,k,l$.

We now have to address that case $j\neq k$ and the case $j=k$ separately.

\subsection{Asymptotics of $\covko_{j,k}$ for $j\neq k$}

Let us now consider $W_0$ as the representative of all other $W_i$ for $i=0,\dots,N_n-1$, written as $W_0=\sum_{l=0}^{M_n-1}X_{n,l}$ with $X_{n,l}=\frac{1}{n}\nu_{j,l}\lambdaok_l \nu_{k,l}$. All the normal random variables involved in such a sum are asymptotically independent but this is not true for the $\lambdaok_l$ since the eigenvalues are constrained to sum to $n$.

Hence, $W_0$ is a triangular array of row-dependent random variables whose asymptotic behaviour can be analyzed by means of \cite[Theorem 2.1]{Neumann_PS2013} that is essentially a Lindeberg-Feller Central Limit Theorem with the row-wise independence relaxed to asymptotic row-wise incorrelation. To analyze the asymptotics of $W_0$ we note that

\[
\Exp[X_{n,l}]=\Exp[\lambdako_l]\Exp[\nu_{j,l}]\Exp[\nu_{k,l}]=0	
\]

Note also that though not independent, the covariance and the correlation between $X_{n,l'}$ and $X_{n,l''}$ is 

\begin{eqnarray*}
\lefteqn{\Exp\left[X_{n,l'}X_{n,l''}\right]=}\\
&=&\frac{1}{n^2}
\Exp\left[\lambdako_{l'}\lambdako_{l''}\right]\Exp\left[\nu_{j,l'}\nu_{j,l''}\right]\Exp\left[\nu_{k,l'}\nu_{k,l''}\right]\\
&=&\frac{1}{n^2}\begin{cases}
0 & \text{if $l'\neq l''$}\\
\Exp[(\lambdako_l)^2] & \text{if $l'=l''=l$}
\end{cases}
\end{eqnarray*}

With this we also know that

\begin{eqnarray*}
\lefteqn{\sigma_{W_0}^2=
\Exp\left[W_i^2\right]=\sum_{l'=0}^{M_n-1}\sum_{l''=0}^{M_n-1}\Exp[X_{n,l'}X_{n,l''}]=}\\
&=&\sum_{l=0}^{M_n-1}\Exp[X_{n,l}^2]=\frac{1}{n^2}\sum_{l=0}^{M_n-1}\Exp[(\lambdako_l)^2]
\end{eqnarray*}

To compute the last expectation we may resort to Lemma \ref{lem:lambdaExp} that gives

\[
\Exp[(\lambdako_l)^2]=
\frac{n-1}{n^{n-1}}\int_0^n p^2(n-p)^{n-2}\dd p=\frac{2n}{n+1}
\]

Hence we have $\sigma^2_{W_0}=\frac{2M_n}{n(n+1)}\rightarrow 0$ for $n\rightarrow\infty$.

This helps satisfying the Lindeberg condition since, if 
for a given $\epsilon>0$ we indicate with
$\Exp\left[X_{n,l}^2\big|_{|X_{n,l}|\ge \epsilon}\right]$ the expectation of $X_{n,l}^2$ restricted to its values that are not less than $\epsilon$ in modulus, then

\[
\sum_{l=0}^{n-1}
\Exp\left[X_{n,l}^2\big|_{|X_{n,l}|\ge \epsilon}\right]\le
\sum_{l=0}^{n-1}
\Exp\left[
X_{n,l}^2
\right]=\sigma^2_{W_0}=\frac{2M_n}{n(n+1)}
\] 

\noindent that vanishes asymptotically.

Finally, let $L_{M_n}, R_{M_n}\subset\{0,\dots,M_n-1\}$ be two index subsets such that $L_{M_n}\cap R_{M_n}=\emptyset$. If $g_{L_{M_n}}$ is any function of the random variables $X_{n,l}$ with $l\in L_{M_n}$ and $h_{R_{M_n}}=\prod_{l\in R_n}X_{n,l}$, the covariance between $g_{L_{M_n}}$ and $h_{R_{M_n}}$ is 

\begin{eqnarray*}
\lefteqn{\Exp\left[g_{L_{M_n}}h_{R_{M_n}}\right]=
\Exp\left[
g_{L_m}
\prod_{l\in R_n}
\frac{1}{n}
\lambdako_l\nu_{j,l}\nu_{k,l}
\right]}\\
&=&\Exp\left[
g_{L_m}
\prod_{l\in R_n}
\lambdako_l
\right]
\prod_{l\in R_n}
\frac{1}{n}
\Exp\left[
\nu_{j,l}
\right]
\Exp\left[
\nu_{k,l}
\right]
=0
\end{eqnarray*}

\noindent that is enough to satisfy the assumptions in equation (2.3) and in equation (2.4) of \cite[Theorem 2.1]{Neumann_PS2013}. From that Theorem, we finally now that $W_0\sim\Gauss{0}{\frac{2M_n}{n(n+1)}}$ when $n\rightarrow\infty$ where convergence is in probability and thus also in distribution. Clearly, the same happens to any $W_i$ for $i=0,\dots,N_n-1$, while $W_{N_n}\sim\Gauss{0}{\frac{2m_n}{n(n+1)}}$.

Let us now consider

\begin{eqnarray*}
\lefteqn{\Pr\left\{\left|\covko_{j,k}\right|\le\epsilon\right\}=}\\
&=& \Pr\left\{\left|\sum_{i=0}^{N_n-1}W_i+W_{N_n}\right|\le\epsilon\right\}\\
&\ge&\Pr\left\{\sum_{i=0}^{N_n-1}\left|W_i\right|+\left|W_{N_n}\right|\le\epsilon\right\}\\
&\ge& \Pr\left\{\max\left\{\left|W_0\right|,\dots,\left|W_{N_n}\right|\right\}\le\frac{\epsilon}{ N_n + 1}\right\}\\
&=& 1-\Pr\left\{\max\left\{\left|W_0\right|,\dots,\left|W_{N_n}\right|\right\}>\frac{\epsilon}{ N_n + 1}\right\}
\end{eqnarray*}

Yet

\begin{eqnarray*}
\lefteqn{\Pr\left\{\max\left\{\left|W_0\right|,\dots,\left|W_{N_n}\right|\right\}>\frac{\epsilon}{ N_n + 1}\right\}}\\
&=&\Pr\left\{\left|W_0\right|>\frac{\epsilon}{ N_n + 1}\,\vee\,\dots\,\vee\,\left|W_{N_n}\right|>\frac{\epsilon}{ N_n + 1}\right\}\\
&\le& (N_n + 1)\Pr\left\{\left|W_0\right|>\frac{\epsilon}{ N_n + 1}\right\}
\end{eqnarray*}

\noindent so that

\[
\Pr\left\{\left|\covko_{j,k}\right|\le\epsilon\right\}\ge
1-(N_n + 1)\Pr\left\{\left|W_0\right|>\frac{\epsilon}{ N_n + 1}\right\}
\]

From this and from the asymptotic normality of $W_0$ we may say

\begin{equation}
\label{eq:jneqk}
\Pr\left\{\left|\covko_{j,k}\right|\le\epsilon\right\}\ge
1-\frac{ N_n + 1}{2}{\rm erfc}\left(\frac{\epsilon}{2\sqrt{\frac{(N_n + 1)^2 M_n}{n(n+1)}}}\right)
\end{equation}

\noindent in which the probability tends to $1$  for $n\rightarrow\infty$, and $M_n=o\left(\frac{n}{\log n}\right)$ chosen to have $N_n=o\left(n^{\nicefrac{1}{2}}\right)$.

\subsection{Asymptotics of $\covko_{j,k}$ for $j=k$}

In this case, exploiting the fact that $\sum_{j=0}^{n-1}\lambdako_j=n$, we may write

\begin{eqnarray*}
\lefteqn{\covok_{j,j}=
1+\sum_{i=0}^{N_n-1}
\underbrace{
\sum_{l=iM_n}^{(i+1)M_n-1}\frac{1}{n}\left(\nu^2_{j,l}-1\right)\lambdako_l}_{W_i}+}\\
&&
\hspace{30mm}
\underbrace{
\sum_{l=n-m_n}^{n-1}\frac{1}{n}\left(\nu^2_{j,l}-1\right)\lambdako_l}_{W_{N_n}}
\end{eqnarray*}

\noindent in which the summands $W_i$ are a triangular arrays of elements with features similar to the previous ones.

In fact, we may focus on $W_0=\sum_{l=0}^{M_n-1}X_{n,l}$ with 
$X_{n,l}=\frac{1}{n}\lambdako_l \left(\nu_{j,l}^2-1\right)$ and note that $\Exp[X_{n,l}]=0$ and 

\begin{eqnarray*}
\lefteqn{\Exp\left[X_{n,l'}X_{n,l''}\right]=}\\
&=&\frac{1}{n^2}
\Exp\left[\lambdako_{l'}\lambdako_{l''}\right]\Exp\left[\left(\nu_{j,l'}^2-1\right)\left(\nu_{j,l''}^2-1\right)\right]\\
&=&
\frac{2}{n^2}
\begin{cases}
0 & \text{if $l'\neq l''$}\\
\Exp[(\lambdako_l)^2] & \text{if $l'=l''=l$}
\end{cases}
\end{eqnarray*}

\noindent where we have exploited tha fact that $\Exp[(\nu_{j,l}^2-1)^2]=\Exp[\nu_{j,l}^4]-2\Exp[\nu_{j,l}^2]+1=2$. With this, $\sigma^2_{W_0}=\frac{4M_n}{n(n+1)}\rightarrow 0$ for $n\rightarrow\infty$.

As before, this makes the Lindeberg condition automatically satisfied and is also enough to satisfy the covariance constraints in equations (2.3) and (2.4) of \cite[Theorem 2.1]{Neumann_PS2013} from which we get that $W_0\sim\Gauss{0}{\frac{4M_n}{n(n+1)}}$ for $n\rightarrow\infty$ where convergence is in probability thus also in distribution. An analogous path leads to the asymptotic behaviour $W_{N_n}\sim\Gauss{0}{\frac{4m_n}{n(n+1)}}$.

The same inequalities as before lead to 

\begin{equation}
\label{eq:jeqk}
\Pr\left\{\left|\covko_{j,j}-1\right|\le\epsilon\right\}\ge
1-\frac{ N_n + 1}{2}{\rm erfc}\left(\frac{\epsilon}{4\sqrt{\frac{ (N_n + 1)^2 M_n}{n(n+1)}}}\right)
\end{equation}

\noindent in which the probability tends to $1$  for $n\rightarrow\infty$, and $M_n=o\left(\frac{n}{\log n}\right)$ chosen to have $N_n=o\left(n^{\nicefrac{1}{2}}\right)$.

From \eqref{eq:jneqk} and \eqref{eq:jeqk} we finally get that $\Uko$ tends to $I_n$ in probability as $n\rightarrow\infty$.
\end{proof}

\begin{proof}[Proof of Lemma \ref{lem:lambdaExp}]

For any function $f:\RR\mapsto\RR$ we have

\begin{eqnarray*}
\lefteqn{\II[f(p)]=
	\int_{\PP^n} f(p_0)\dd p_0\dots\dd p_{n-1}}\\
	&=&
	\int_0^n f(p_0)
	\int_0^{n-p_0}
	\int_0^{n-p_0-p_1}
	\!\!\!\!\!\!
	\dots
	\int_0^{n-p_0-p_1-\dots-p_{n-3}}
	\!\!\!\!\!\!\!\!\!\!\!\!\!\!\!\!\!\!\!\!\!\!
	\dd p_0\dots\dd p_{n-2}\\
	&=&
	\int_0^n f(p_0)\frac{(n-p_0)^{n-2}}{(n-2)!}\dd p_0
\end{eqnarray*}

Since $\lambdako$ is uniformly distributed over $\PP^n$ the probability density is the constant $1/\II[1]=n^{-(n-1)}(n-1)!$ and the expectation of $f$ is

\begin{eqnarray*}
\lefteqn{\Exp[f(\lambdaok_j)]=n^{-(n-1)}(n-1)!\II[f(p)]}\\
&=& \frac{(n-1)!}{n^{n-1}}\int_0^n f(p)\frac{(n-p)^{n-2}}{(n-2)!}\dd p\\
&=&
\frac{n-1}{n^{n-1}}\int_0^n f(p)(n-p)^{n-2}\dd p
\end{eqnarray*}
\end{proof}

\bibliographystyle{IEEEtran}
\bibliography{reference}

\end{document}